\documentclass[11pt,letter]{article}
\usepackage{amsfonts,amssymb,enumitem}
\usepackage{amsmath,amsthm,color}
\usepackage[T1]{fontenc}
\usepackage{lmodern}
\usepackage[utf8]{inputenc}
\usepackage{hyperref,cite,fullpage}
\usepackage[capitalise]{cleveref}
\usepackage{xspace}
\usepackage{enumitem}
\usepackage{calc}

\newcommand{\Bag}{X}
\newcommand{\tw}{\mathrm{tw}}
\newcommand{\tin}{\mathsf{tree} \textnormal{-} \alpha}
\newcommand{\HH}{\mathcal{H}}

\newcommand{\Q}{\mathbb{Q}}
\renewcommand{\P}{\textsf{P}}
\newcommand{\NP}{\textsf{NP}}

\renewcommand{\O}{\mathcal{O}}
\newcommand{\dist}{\textsf{dist}}

\newcommand{\Oh}{\O}
\newcommand{\forked}[1]{\Breve{#1}}
\newcommand{\msotwo}{$\mathsf{MSO}_2$\xspace}
\newcommand{\cmsotwo}{$\mathsf{CMSO}_2$\xspace}
\newcommand{\cpmsotwo}{$\mathsf{C}_{\leq p}\mathsf{MSO}_2$\xspace}

\newcommand{\msotypes}{\mathsf{Types}}
\newcommand{\msotype}{\mathsf{type}}
\newcommand{\wei}{\mathfrak{w}}

\newcommand{\dom}{\mathrm{dom}}
\newcommand{\cT}{\mathcal{T}}
\newcommand{\Tab}{\mathsf{Tab}\xspace}

\newtheorem{theorem}{Theorem}[section]

\newtheorem{corollary}[theorem]{Corollary}
\newtheorem{proposition}[theorem]{Proposition}
\newtheorem{lemma}[theorem]{Lemma}
\newtheorem{observation}[theorem]{Observation}

\theoremstyle{definition}
\newtheorem{remark}[theorem]{Remark}

\title{Tree decompositions with bounded independence number:\\
beyond independent sets}

\author{Martin Milani{\v c}\\
\small FAMNIT and IAM, University of Primorska, Koper, Slovenia\\
\small \texttt{martin.milanic@upr.si}\\
\and
Pawe\l{} Rz\k{a}\.zewski\\
\small  Warsaw University of Technology and University of Warsaw, Poland\\
\small \texttt{pawel.rzazewski@pw.edu.pl}
}

\begin{document}

\maketitle

\begin{abstract}
We continue the study of graph classes in which the treewidth can only be large due to the presence of a large clique, and, more specifically, of graph classes with bounded tree-independence number.
In [Dallard, Milani\v{c}, and \v{S}torgel, Treewidth versus clique number. {II}. Tree-independence number, 2022], it was shown that the Maximum Weight Independent Packing problem, which is a common generalization of the Independent Set and Induced Matching problems, can be solved in polynomial time provided that the input graph is given along with a tree decomposition with bounded independence number.
We provide further examples of algorithmic problems that can be solved in polynomial time under this assumption.
This includes, for all even positive integers $d$,  the problem of packing subgraphs at distance at least $d$ (generalizing the Maximum Weight Independent Packing problem) and the problem of finding a large induced sparse subgraph satisfying an arbitrary but fixed property expressible in counting monadic second-order logic.
As part of our approach, we generalize some classical results on powers of chordal graphs to the context of general graphs and their tree-independence numbers.
\end{abstract}


\section{Introduction}\label{sec:intro}
Structured graph decompositions and the corresponding graph width parameters have become one of the central tools for dealing with algorithmically hard graph problems.
One of the most well-known and well-studied graph width parameters is the \emph{treewidth} of a graph, which, roughly speaking, measures how similar the graph is to a tree.
Since graphs of bounded treewidth are necessarily sparse, several generalizations of treewidth were proposed in the literature that allow for families of dense graphs, including clique-width,
rank-width,
Boolean-width,
mim-width,
and twin-width.
A common feature of all these parameters is that when the parameter of the input graph is bounded by a fixed constant, efficient dynamic programming algorithms can be developed for various graph problems, typically applied along some hierarchical decomposition of the graph into simpler parts (see, e.g.,~\cite{MR985145,MR1739644,MR4020556,MR2857670,MR4288865,MR4232070,MR4189412,MR3967291}, as well as the algorithmic metatheorems~\cite{MR1105479,MR3126917,MR3126918,MR1739644}).

Another generalization of bounded treewidth that allows for families of dense graphs is given by the concept of \emph{$(\tw,\omega)$-boundedness}, a property of a graph class meaning that the treewidth of graphs in the class can only be large due to the presence of a large clique, and the same is true for all induced subgraphs of graphs in the class.
While this property is known to imply some good algorithmic properties related to clique and coloring problems (see~\cite{DBLP:journals/endm/ChaplickZ17,MR4332111,DMS-WG2020,dallard2021treewidth,dallard2022firstpaper}), it is an interesting question whether $(\tw,\omega)$-boundedness has any other algorithmic implications, in particular, for problems related to independent sets.
This question was partially answered by Dallard, Milani\v{c}, and \v{S}torgel, who initiated a systematic investigation of $(\tw,\omega)$-bounded graph classes (see~\cite{dallard2021treewidth,DMS-WG2020}, as well as~\cite{dallard2022firstpaper,dallard2022secondpaper}).
In particular, in~\cite{dallard2022firstpaper} a sufficient condition was identified for $(\tw,\omega)$-bounded graph classes to admit a polynomial-time algorithm for the \textsc{Max Weight Independent Packing} problem and, as a consequence, for the weighted variants of the \textsc{Independent Set} and \textsc{Induced Matching} problems. 
These results were obtained using the notion of \emph{tree-independence number}, a graph invariant related to tree decompositions, which is defined similarly to treewidth but aims at minimizing the maximum size of an \emph{independent} set of vertices contained in a bag of a tree decomposition---as opposed to an arbitrary set of vertices, which is the case for treewidth. 

Ramsey's theorem implies that every graph class with bounded tree-independence number is $(\tw,\omega)$-bounded. 
As shown in~\cite{dallard2022secondpaper}, the converse implication is known to hold in graph classes defined by excluding a single graph with respect to any of six well-known graph containment relations (the subgraph, topological minor, and minor relations, as well as their induced variants).
In general, it is not known whether the two concepts are equivalent; it is conjectured that they are~\cite{dallard2022secondpaper}.
The main algorithmic result from~\cite{dallard2022firstpaper} states that the \textsc{Max Weight Independent Packing} problem can be solved in polynomial time provided that the input graph is given along with a tree decomposition with bounded independence number.
(The precise definition of the problem and the result statement will be given in \cref{sec:prelim}.)
Several questions regarding the computation of tree-independence number left open by the work~\cite{dallard2022firstpaper} were resolved in a subsequent work by Dallard et al.~\cite{dallard2022computing}, where the following results were obtained.
First, it was shown that the exact computation of the tree-independence number is \textsf{para-NP-hard}, more precisely, that for every $k\ge 4$ it is $\NP$-complete to determine if a given graph has tree-independence number at most $k$. 
Second, a polynomial-time algorithm was developed for computing a tree decomposition with bounded independence number, whenever one exists.
Given an $n$-vertex graph $G$ and an integer $k$, in time $2^{\mathcal{O}(k^2)}n^{\mathcal{O}(k)}$ the algorithm either outputs a tree decomposition of $G$ with independence number at most $8k$, or determines that the tree-independence number of $G$ is larger than $k$.

\begin{sloppypar}
Combined with the results from~\cite{dallard2022firstpaper}, it follows that the \textsc{Max Weight Independent Packing} problem is polynomial-time solvable in any class of graphs with bounded tree-independence number.
Thus, boundedness of the tree-independence number is an algorithmically useful refinement of \hbox{$(\tw,\omega)$-boundedness} that is still general enough to capture a wide variety of families of graph classes, including, among others, graph classes of bounded treewidth, graph classes of bounded independence number, and intersection graphs of connected subgraphs of graphs with bounded treewidth (see~\cite[Theorem 3.12]{dallard2022firstpaper}).
\end{sloppypar}

\begin{sloppypar}
In this paper,
we provide further examples for algorithmic usefulness of bounded tree-independence number.
\end{sloppypar}


\subsection{Our results and relation to previous work}

We obtain two main sets of results.

\begin{sloppypar}
First, we generalize the polynomial-time solvability of the \textsc{Max Weight Independent Packing} problem in classes of graphs equipped with a tree decomposition with bounded tree-independence number from~\cite{dallard2022firstpaper} to the problem of packing subgraphs at distance $d$, for all even positive integers $d$.
The case $d = 2$ corresponds to the result from~\cite{dallard2022firstpaper}.
Unless $\P = \NP$, this result does not generalize to odd values of $d$, since  the distance-$3$ variant of the \textsc{Independent Set} problem is $\NP$-hard for chordal graphs, which have tree-independence number at most one.
\end{sloppypar}

We obtain this algorithmic result by proving a structural result on tree-independence number that may be of interest on its own.
Given a graph $G$ and a positive integer $k$, we denote by $G^k$ the \emph{$k$-th power} of $G$, that is, the graph obtained from the graph $G$ by adding to it all edges between pairs of vertices at distance at most $k$.
We show that for any graph $G$ and any positive integer $k$, the tree-independence number of the graph $G^{k+2}$ cannot exceed that of the graph $G^k$.
This is a significant generalization of a result of Duchet~\cite{MR778751}, who proved the same result for the case when $G^k$ is a chordal graph (that is, when $G^k$ has tree-independence number at most one).
As a consequence, we obtain that for every positive integer $k$, the class of graphs with tree-independence number at most $k$ is closed under taking odd powers, which generalizes the analogous result for the class of chordal graphs proved by Balakrishnan and Paulraja~\cite{MR704427}.
We complement this result by observing that the class of even powers of chordal graphs (for any fixed power) is not contained in any nontrivial hereditary graph class.
This implies in particular that any such class is $(\tw,\omega)$-unbounded and has unbounded tree-independence number.

\medskip
Our second set of results is an algorithmic metatheorem for graphs equipped with a tree decomposition with bounded independence number.
It deals with the problem of finding a large induced subgraph of bounded chromatic number, satisfying an arbitrary but fixed property expressible in Counting Monadic Second-Order logic (\cmsotwo). This far-reaching generalization of \textsc{Max Weight Independent Set} goes along the lines of analogous results proven in various settings, e.g., for bounded-treewidth graphs~\cite{DBLP:journals/iandc/Courcelle90}, for graphs with polynomially many \emph{potential maximal cliques}~\cite{DBLP:journals/siamcomp/FominTV15}, or for graphs with no long induced cycles~\cite{DBLP:conf/soda/AbrishamiCPRS21,DBLP:conf/stoc/GartlandLPPR21}.
An astute reader might notice that usually such problems concern finding large \emph{sparse} subgraphs, which typically means bounded treewidth~\cite{DBLP:journals/siamcomp/FominTV15,DBLP:conf/soda/AbrishamiCPRS21} or degeneracy~\cite{DBLP:conf/stoc/GartlandLPPR21}.
However, it is straightforward to observe that in classes of bounded tree-independence number treewidth is bounded by a linear function of the chromatic number. Note that this is far from the truth in general graphs; consider, e.g., large complete bipartite graphs.

Some examples of problems that can be solved in polynomial time with this approach include
\begin{itemize}
    \item finding the largest induced forest (which is equivalent to \textsc{Min Feedback Vertex Set}),
    \item finding the largest induced bipartite subgraph (which is equivalent to \textsc{Min Odd Cycle Transversal}),
    \item finding the largest induced planar subgraph (which is equivalent to \textsc{Planarization}),
    \item finding the maximum number of pairwise disjoint and non-adjacent cycles,
    \item given a graph whose every vertex $v$ is equipped with a list $L(v) \subseteq \{1,\ldots,r\}$ for constant $r$,
    finding a largest induced subgraph that admits a proper $r$-coloring respecting lists $L$.
\end{itemize}

We point out that the bound on the chromatic number of the subgraph we are looking for must be constant. 
Indeed, \textsc{Coloring} is \NP-hard in graphs with tree-independence number 2, and \textsc{List Coloring} is \NP-hard in graphs with tree-independence number 1 (i.e., chordal graphs).

\subsection{Organization of the paper}

In \cref{sec:prelim}, we provide the necessary preliminaries.
In \cref{sec:distance}, we consider the tree-independence number in the context of graph powers, and develop algorithmic implications of the structural results for the problem of packing connected subgraphs at even distance.
In \cref{sec:cmso}, we develop a general theorem about polynomial-time solvability of the problem of finding a large induced subgraph of small chromatic number satisfying an arbitrary but fixed property expressible in \cmsotwo logic, provided that the input graph is equipped with a tree decomposition with bounded independence number.

\section{Preliminaries}\label{sec:prelim}
For an integer $n$, by $[n]$ we denote the set $\{1,\ldots,n\}$.
Let $G$ be a graph. 
For a set $X \subseteq V(G)$, by $G[X]$ we denote the subgraph of $G$ induced by $X$.
A \emph{clique} in a graph $G$ is a set of pairwise adjacent vertices; an \emph{independent set} is a set of pairwise non-adjacent vertices.

The \emph{distance} between two vertices $u$ and $v$ in $G$ is denoted by $\dist_G(u,v)$ and defined as the length of a shortest path between $u$ and $v$ (or $\infty$ if there is no $u$-$v$-path in $G$).
Given a positive integer $k$, the \emph{$k$-th power} of $G$ is the graph $G^k$ with vertex set $V$, in which two distinct vertices $u$ and $v$ are adjacent if and only if $\dist_G(u,v)\le k$. Note that $G^1 =G$.

A graph is \emph{chordal} if it does not contain any induced cycles of length at least four.
A graph is \emph{split} if it admits a partition of the vertex set into a clique and an independent set.
It is known (and easy to see) that every split graph is a chordal graph.

\paragraph{Tree decompositions and tree-independence number.}
A \emph{tree decomposition} of a graph $G$ is a pair $\mathcal{T} = (T, \{X_t\}_{t\in V(T)})$ where $T$ is a tree and every node $t$ of $T$ is assigned a vertex subset $X_t\subseteq V(G)$ called a \emph{bag} such that the following conditions are satisfied: every vertex is in at least one bag, for every edge $uv\in E(G)$ there exists a node $t\in V(T)$ such that $X_t$ contains both $u$ and $v$, and for every vertex $u\in V(G)$ the subgraph of $T$ induced by the set $\{t\in V(T): u\in X_t\}$ is connected (that is, a tree).

Usually we consider tree decompositions to be rooted. 
For a node $t$ of $T$, by $T_t$ we denote the subtree of $T$ rooted at $t$, and by $G_t$ we denote the subgraph of $G$ induced by $\bigcup_{t' \in V(T_t)} X_{t'}$.
A rooted tree decomposition $\cT = (T,\{X\}_{t \in V(T)})$ of a graph $G$ is \emph{nice}
if each node $t$ is of one of the following types:
\begin{description}[leftmargin=!,labelwidth=\widthof{\bfseries introduce}]
\item[leaf] a leaf in the decomposition tree such that $|X_t|=1$,
\item[introduce] an inner node with a single child $t'$, such that $X_t = X_{t'} \cup \{v\}$ for some $v \in V(G) \setminus X_{t'}$,
\item[forget] an inner node with a single child $t'$, such that $X_t = X_{t'} \setminus \{v\}$ for some $v \in X_{t'}$,
\item[join] an inner node with exactly two children $t',t''$, such that $X_t = X_{t'} = X_{t''}$.
\end{description}
It is well known that given a tree decomposition $\cT=(T,\{X_t\}_{t \in V(T)})$ of $G$, one can obtain in polynomial time a nice tree decomposition $\cT'=(T',\{X't\}_{t \in V(T')})$ of $G$ with at most $\max_{t \in V(T)} |X_t| \cdot |V(G)|$ nodes and whose every bag is contained in some bag of $\cT$ (see, e.g., the textbook~\cite{MR3380745}, or~\cite{dallard2022firstpaper} for a treatment focused on the independence number).

Consider a graph $G$ and a tree decomposition $\mathcal{T} = (T, \{X_t\}_{t\in V(T)})$ of $G$.
The \emph{independence number} of $\mathcal T$, denoted by $\alpha(\mathcal{T})$, is defined as follows:
\[\alpha(\mathcal{T}) = \max_{t\in V(T)} \alpha(G[X_t])\,.\]

Note that if $\cT$ and $\cT'$ are tree decompositions such that every bag of $\cT'$ is contained in some bag of $\cT$ (in particular, 
if $\cT'$ is a nice tree decomposition obtained from $\cT$), then we have $\alpha(\cT') \leq \alpha(\cT)$.

The \emph{tree-independence number} of $G$, denoted by $\tin(G)$, is the minimum independence number among all possible tree decompositions of $G$.
A graph $G$ is \emph{chordal} if it has no induced subgraphs isomorphic to a cycle of length at least four.
A graph $G$ is chordal if and only if $\tin(G)\le 1$ (see~\cite{dallard2022firstpaper}).

\paragraph{Independent packing of subgraphs.}
Given a graph $G$ and a family $\HH =\{H_j\}_{j\in J}$ of subgraphs of $G$, we denote by $G(\HH)$ the graph with vertex set $J$, in which two distinct elements $i,j\in J$ are adjacent if and only if $H_i$ and $H_j$ either have a vertex in common or there is an edge in $G$ connecting them.
This construction was considered by Cameron and Hell in~\cite{MR2190818}, who focused on the particular case when $\HH$ is the set of all subgraphs of $G$ isomorphic to a member of a fixed family $\mathcal{F}$ of connected graphs; they showed that if $G$ is a chordal graph, then so is $\mathcal{H}(G)$.
Dallard et al.~generalized this result as follows.

\begin{sloppypar}
\begin{lemma}[Dallard et al.~\cite{dallard2022firstpaper}]\label{tin-of-G(H)}
Let $G$ be a graph, let \hbox{$\mathcal{T} = (T, \{X_t\}_{t\in V(T)})$} be a tree decomposition of $G$, and let $\HH =\{H_j\}_{j\in J}$ be a finite family of connected non-null subgraphs of $G$.
Then \hbox{$\mathcal{T}' = \big(T, \{\Bag'_t\}_{t\in V(T)}\big)$} with $\Bag'_t = \{j\in J : V(H_j) \cap X_t \neq \emptyset\}$ for all $t \in V(T)$ is a tree decomposition of $G(\HH)$ such that $\alpha(\mathcal{T}')\le \alpha(\mathcal{T})$.
\end{lemma}
\end{sloppypar}

Let us now explain the main algorithmic result related to packing subgraphs from~\cite{dallard2022firstpaper}.
Given a graph $G$ and a family $\HH =\{H_j\}_{j\in J}$ of subgraphs of $G$, a subfamily $\HH'$ of $\HH$ is said to be an \emph{independent $\HH$-packing} in $G$ if every two graphs in $\HH'$ are vertex-disjoint and there is no edge between them, in other words, if $\HH'$ is an independent set in the graph $G(\HH)$.
Assume now that the subgraphs in $\HH$ are equipped with a weight function $\wei:J\to \mathbb{Q}_+$ assigning weight $\wei_j$ to each subgraph $H_j$.
For any set $I\subseteq J$, we define the \emph{weight} of the family $\HH' = \{H_i\}_{i\in I}$ as the sum $\sum_{i\in I}\wei_i$.
In particular, the weight of any independent $\HH$-packing is well-defined.
Given a graph $G$, a finite family $\HH = \{H_j\}_{j\in J}$ of connected non-null subgraphs of $G$, and a weight function $\wei:J\to \mathbb{Q}_+$ on the subgraphs in $\HH$, the \textsc{Max Weight Independent Packing} problem asks to find an independent $\HH$-packing in $G$ of maximum weight.

When $\HH$ is the set of all subgraphs of $G$ isomorphic to a member of a fixed family $\mathcal{F}$ of connected graphs, the problem coincides with the \textsc{Max Weight Independent $\mathcal{F}$-Packing} problem studied in Dallard et al.~\cite{dallard2022firstpaper}.
This latter problem in turn generalizes several problems studied in the literature, including  the \textsc{Independent $\mathcal{F}$-Packing} problem (see~\cite{MR2190818}),  which corresponds to the unweighted case, the \textsc{Max Weight Independent Set} problem, which corresponds to the case $\mathcal{F} = \{K_1\}$, the \textsc{Max Weight Induced Matching} problem (see, e.g.,~\cite{MR3776983,MR4151749}), which corresponds to the case $\mathcal{F} = \{K_2\}$, and the \textsc{Dissociation Set} problem (see, e.g.,~\cite{MR3593941,MR615221,MR2812599}), which corresponds to the case when $\mathcal{F}= \{K_1,K_2\}$ and the weight function assigns to each subgraph $H_j$ the weight equal to $|V(H_j)|$, and the \textsc{$k$-Separator} problem (see, e.g.,~\cite{DBLP:conf/soda/Lee17,MR3987192,MR3296270}), which corresponds to the case when $\mathcal{F}$ contains all connected graphs with at most $k$ vertices, the graph $G$ is equipped with a vertex weight function $\wei:V(G)\to \mathbb{Q}_+$, and the weight function on $\HH$ assigns to each subgraph $H_j$ the weight equal to $\sum_{x\in V(H_j)}\wei(x)$.

Dallard et al.~showed in~\cite{dallard2022firstpaper} that the \textsc{Max Weight Independent Packing} problem can be solved in polynomial time for graphs equipped with a tree decomposition with bounded independence number.

\begin{sloppypar}
\begin{theorem}[Dallard et al.~\cite{dallard2022firstpaper}]\label{max-weight-independent-subgraph-packing-for-bounded-tree-alpha}
Fix an integer $k \geq 1$.
Then, given a graph $G$ and a finite family $\mathcal{H} = \{H_j\}_{j \in J}$ of connected nonnull subgraphs of $G$, the \textsc{Max Weight Independent Packing} problem can be solved in time \hbox{$\mathcal{O}(|J| \cdot ((|J| + |V(T)|) \cdot |V(G)| + |E(G)|+|J|^{k}\cdot|V(T)|)$} if $G$ is given together with a tree decomposition \hbox{$\mathcal{T} = (T, \{\Bag_t\}_{t\in V(T)})$} with independence number at most $k$.
\end{theorem} 
\end{sloppypar}

\section{Packing subgraphs at even distance}\label{sec:distance}
Given a positive integer $d$, a \emph{distance-$d$ independent set} in a graph $G$ is a set of vertices at pairwise distance at least $d$, that is, an independent set in the graph $G^{d-1}$.
The \textsc{Distance-$d$ Independent Set} problem is the following problem:
Given a graph $G$ and a weight function $\wei : V(G) \to \Q_+$, find a distance-$d$ independent set with maximum total weight.
For $d = 2$, the problem coincides with the classical \textsc{Max Weight Independent Set} problem.

\subsection{Tree-independence number and graph powers}

Balakrishnan and Paulraja proved in~\cite{MR704427} that the class of chordal graphs is closed under taking odd powers.
Combined with the polynomial-time solvability of the \textsc{Max Weight Independent Set} problem in the class of chordal graphs (see, e.g.,~\cite{MR0392683}), this implies that the \textsc{Distance-$d$ Independent Set} problem is solvable in polynomial time in the class of chordal graphs for all even $d\ge 2$.
On the other hand, Eto et al.~\cite{MR3149106} showed that the problem is \NP-complete in the class of chordal graphs for all odd $d\ge 3$.

The result of Balakrishnan and Paulraja was generalized by Duchet as follows.

\begin{theorem}[Duchet~\cite{MR778751}]\label{chordal-from-k-to-k+2}
Let $G$ be a graph and $k$ a positive integer such that $G^k$ is chordal.
Then the graph $G^{k+2}$ is also chordal.
\end{theorem}

In terms of tree-independence number, \cref{chordal-from-k-to-k+2} can be phrased as follows.

\begin{theorem}
Let $G$ be a graph and $k$ a positive integer such that $\tin(G^k)\le 1$.
Then $\tin(G^{k+2})\le 1$.
\end{theorem}

We now generalize this result as follows.

\begin{theorem}\label{tin-from-k-to-k+2}
Let $G$ be a graph and $k$ a positive integer.
Then $\tin(G^{k+2})\le \tin(G^k)$.
\end{theorem}

The proof of \cref{tin-from-k-to-k+2} will be based on the following lemma, the case $d = 1$ of which was observed by Duchet in~\cite{MR778751}.

\begin{sloppypar}
\begin{lemma}\label{equality-between-G^k+2d-and-G^k(H)}
Let $G=(V,E)$ be a graph, and $k$ and $d$ positive integers.
For $v\in V$, let $H_v$ be the subgraph of $G$ induced by the vertices at distance at most $d$ from $v$, and let $\mathcal{H} = \{H_v\}_{v\in V}$.
Then $G^{k+2d} = G^k(\mathcal{H})$.
\end{lemma}
\end{sloppypar}

\begin{proof}
Let us first observe that for every vertex $v$ of $G$, the graph $H_v$ is also a subgraph of $G^k$.
Hence, the graph $G^k(\mathcal{H})$ is well-defined.
The graphs $G^{k+2d}$ and $G^k(\mathcal{H})$ both have the same vertex set, namely $V$.

As for the edge set, consider two distinct vertices $u$ and $v$ in $V$. 
Assume first that $u$ and $v$ are adjacent in $G^k(\mathcal{H})$.
Then the graphs $H_u$ and $H_v$ either have a vertex in common, or they are vertex-disjoint but there is an edge between them in $G^k$.
If $H_u$ and $H_v$ have a vertex in common, say $x\in V(H_u)\cap V(H_v)$, then $\dist_G(u,v)\le \dist_G(u,x)+ \dist_G(x,v) \le 2d\le k+2d$, and hence $u$ and $v$ are adjacent in $G^{k+2d}$.
If $H_u$ and $H_v$ are vertex-disjoint but there exist vertices $x\in V(H_u)$ and $y\in V(H_v)$ such that $\{x,y\}\in E(G^k)$, then $\dist_G(u,v)\le \dist_G(u,x)+ \dist_G(x,y) + \dist_G(y,v)\le d+k+d= k+2d$, and hence $u$ and $v$ are adjacent in $G^{k+2d}$.

Assume now that $u$ and $v$ are adjacent in $G^{k+2d}$.
Then $\dist_G(u,v)\le k+2d$.
We want to show that the graphs $H_u$ and $H_v$ either have a vertex in common, or they are vertex-disjoint but there is an edge between them in $G^k$.
If $u$ and $v$ at distance at most $2d$ in $G$, then there exists a vertex $x\in V$ that is at distance in $G$ at most $d$ from each of $u$ and $v$, and hence the graphs $H_u$ and $H_v$ have a vertex in common.
We may thus assume that $\dist_G(u,v)\ge 2d+1$ or, equivalently, that the graphs $H_u$ and $H_v$ are vertex-disjoint.
Since $2d+1\le \dist_G(u,v)\le k+2d$, there exists a $u$-$v$-path $P = (u = u_0,u_1,\ldots, u_q = v)$ in $G$ such that $2d+1\le q\le k+2d$.
This implies that $1\le \dist_G(u_d,u_{q-d})\le k$ and hence the vertices $u_d$ and $u_{q-d}$ are adjacent in $G^k$.
Since $u_d$ is at distance at most $d$ from $u$, it belongs to $H_u$.
Similarly, $u_{q-d}$ belongs to $H_v$.
Thus, $\{u_d,u_{q-d}\}$ is an edge in $G^k$ between the subgraphs $H_u$ and $H_v$.
This shows that $u$ and $v$ are adjacent in $G^k(\mathcal{H})$.

We showed that the graphs $G^{k+2d}$ and $G^k(\mathcal{H})$ have the same vertex sets and the same edge sets. 
Hence, they are the same.
\end{proof}

\begin{proof}[Proof of \cref{tin-from-k-to-k+2}.]
Let $G=(V,E)$ be a graph and $k$ a positive integer.
For $v\in V$, let $H_v$ be the subgraph of $G$ induced by $v$ and its neighbors, and let $\mathcal{H} = \{H_v\}_{v\in V}$.
Let $q = \tin(G^k)$ and let \hbox{$\mathcal{T} = (T, \{X_t\}_{t\in V(T)})$} be a tree decomposition of $G^k$ such that the independence number of $\mathcal{T}$ equals $q$.
Let \hbox{$\mathcal{T}' = (T, \{\Bag'_t\}_{t\in V(T)})$} with $\Bag'_t = \{v\in V : V(H_v) \cap X_t \neq \emptyset\}$ for all $t \in V(T)$.
Since for each $v\in V$, the graph $H_v$ is a connected non-null subgraph of $G^k$, \cref{tin-of-G(H)} implies that $\mathcal{T}'$ is a tree decomposition of $G^k(\HH)$ such that $\alpha(\mathcal{T}')\le \alpha(\mathcal{T}) = q$.
Since $G^{k+2} = G^k(\mathcal{H})$ by \cref{equality-between-G^k+2d-and-G^k(H)}, we conclude that $\tin(G^{k+2})\le \alpha(\mathcal{T}') \le q = \tin(G^k)$.
\end{proof}

\cref{tin-from-k-to-k+2} has the following consequence.

\begin{corollary}\label{tin-can-only-go-down-when-taking-odd-powers}
Let $G$ be a graph and $k$ an odd positive integer.
Then $\tin(G^{k})\le \tin(G)$.
\end{corollary}

We will also make use of the following algorithmic version of \cref{tin-can-only-go-down-when-taking-odd-powers}.

\begin{sloppypar}
\begin{corollary}\label{tin-can-only-go-down-when-taking-odd-powers-algorithmic}
For every positive odd integer $k$, there exists an algorithm that takes as input a graph $G$ and a tree decomposition \hbox{$\mathcal{T} = (T, \{\Bag_t\}_{t\in V(T)})$} of $G$, and computes in time $\mathcal{O}\left((|V(G)|+|V(T)|)\cdot(|V(G)|+|E(G)|)\right)$ the graph $G^k$ and a tree decomposition \hbox{$\mathcal{T}' = (T, \{\Bag'_t\}_{t\in V(T)})$} of $G^k$ such that $\alpha(\mathcal{T}')\le \alpha(\mathcal{T})$.
\end{corollary}
\end{sloppypar}

\begin{proof}
The statement is trivial for $k = 1$, so let us assume that $k\ge 3$.
Let $d = (k-1)/2$.
For $v\in V(G)$, let $H_v$ be the subgraph of $G$ induced by the vertices at distance at most $d$ from $v$ in $G$, and let $\mathcal{H} = \{H_v\}_{v\in V(G)}$.
\Cref{equality-between-G^k+2d-and-G^k(H)} implies that $G^{k} = G(\mathcal{H})$.
Let \hbox{$\mathcal{T}' = (T, \{\Bag'_t\}_{t\in V(T)})$} with $\Bag'_t = \{v\in V(G) : V(H_v) \cap X_t \neq \emptyset\}$ for all $t \in V(T)$.
Since for each $v\in V(G)$, the graph $H_v$ is a connected non-null subgraph of $G$, \cref{tin-of-G(H)} implies that $\mathcal{T}'$ is a tree decomposition of $G^k$ such that $\alpha(\mathcal{T}')\le \alpha(\mathcal{T})$.
Note that for a vertex $v\in V(G)$, the condition $V(H_v) \cap X_t \neq \emptyset$ is equivalent to the condition that $v$ is at distance at most $d$ in $G$ from $X_t$.
Thus, each bag $\Bag'_t$ can be computed in time $\mathcal{O}(|V(G)|+|E(G)|)$ using a BFS traversal up to distance $d+1$ from a new vertex adjacent to all vertices of the bag $X_t$.
The graph $G^k$ can be computed using a BFS traversal from each vertex $v\in V(G)$ in time 
$\mathcal{O}(|V(G)|(|V(G)|+|E(G)|))$.
 The total time complexity of this approach is $\mathcal{O}\left((|V(G)|+|V(T)|)\cdot(|V(G)|+|E(G)|)\right)$.
\end{proof}


Finally, let us remark that a result similar to \cref{tin-can-only-go-down-when-taking-odd-powers} does not hold for even powers.

\begin{proposition}\label{prop:no-bound}
Fix an even positive integer $k$.
Then there is no function $f$ satisfying 
\[\tin(G^{k})\le f(\tin(G))\] for all graphs $G$.
\end{proposition}

\begin{proof}
Since the tree-independence number of chordal graphs is at most one, it suffices to show that the 
$k$-th powers of chordal graphs have arbitrarily large tree-independence number.
To this end, we show that for every graph $H$ there exists a chordal graph $G$ such that $G^k$ contains an induced subgraph isomorphic to $H$.
This will suffice, since there exist graphs with arbitrarily large tree-independence numbers (for example complete bipartite graphs $K_{n,n}$, see~\cite{dallard2022firstpaper}).
First, let $\widehat H$ be the graph obtained from $H$ by subdividing each edge exactly once and adding edges so that the newly introduced vertices form a clique.
Then $V(H)$ is an independent set in $\widehat H$ and $V(\widehat H)\setminus V(H)$ is a clique in $\widehat H$.
Hence $\widehat H$ is a split graph and thus a chordal graph.
We obtain the graph $G$ from $\widehat H$ by appending to each vertex $v\in V(H)$ a path $P^v$ of length $(k-2)/2$ such that one endpoint of $P^v$ is $v$ and all the other vertices of $P^v$ are new.
In particular, for any two different vertices $v,w\in V(H)$, the corresponding paths $P^v$ and $P^w$ are vertex-disjoint.
By construction, the graph $G$ is chordal.
For each vertex $v\in V(H)$, let us denote the two endpoints of $P^v$ by $v$ and $v'$ (with $v' = v$ if and only if $k = 2$), and let $X = \{v':v\in V(H)\}$.
For arbitrary two distinct vertices $u,v\in V(H)$, it holds that $u'v'\in E(G^k)$ if and only if $\dist_G(u',v')\le k$, which happens if and only if $u$ and $v$ are adjacent in $H$.
Thus, the subgraph of $G^k$ induced by $X$ is isomorphic to $H$.
This shows that for every graph $H$ there exists a chordal graph $G$ such that $G^k$ contains an induced subgraph isomorphic to $H$, as claimed.
\end{proof}

\begin{remark}
Since the balanced complete bipartite graph $K_{n,n}$ has clique number two and treewidth~$n$, the above proof shows, in fact, that for any even positive integer $k$, the class of $k$-th powers of chordal graphs is $(\tw,\omega)$-unbounded.
\end{remark}



\subsection{Algorithmic consequences}

Given a graph $G$ and two sets $A, B \subseteq V(G)$, we denote by $\dist_G(A,B)$ the minimum over all distances in $G$ between a vertex in $A$ and a vertex in $B$.
This definition is naturally extended to subgraphs of $G$, by considering the distance between their vertex sets.

Fix a positive integer $d$.
Given a graph $G$ and a finite family $\HH = \{H_j\}_{j\in J}$ of connected non-null subgraphs of $G$, a \emph{distance-$d$ $\HH$-packing} in $G$ is a subfamily $\HH' = \{H_i\}_{i\in I}$ of subgraphs from $\HH$ (that is, $I\subseteq J$) that are at pairwise distance at least $d$.
If the subgraphs in $\HH$ are equipped with a weight function $\wei:J\to \mathbb{Q}_+$ assigning weight $\wei_j$ to each subgraph $H_j$, we define the \emph{weight} of a distance-$d$ independent $\HH$-packing $\HH'=\{H_i\}_{i\in I}$ in $G$ as the sum $\sum_{i\in I}\wei_i$.
Given a graph $G$, a finite family $\HH = \{H_j\}_{j\in J}$ of connected non-null subgraphs of $G$, and a weight function $w:J\to \mathbb{Q}_+$ on the subgraphs in $\HH$, the \textsc{Maximum Weight Distance-$d$ Packing} problem asks to find a distance-$d$ $\HH$-packing in $G$ of maximum weight.
When $d = 2$ and $\HH$ is the set of all subgraphs of $G$ isomorphic to a member of a fixed family $\mathcal{F}$ of connected graphs, the problem coincides with the \textsc{Maximum Weight Independent $\mathcal{F}$-Packing} problem studied in Dallard et al.~\cite{dallard2022firstpaper}.

\begin{sloppypar}
\begin{observation}\label{observation}
Let $d$ be a positive integer, let $G$ be a graph, let $\HH = \{H_j\}_{j\in J}$ be a finite family of connected non-null subgraphs of $G$, and let $I\subseteq J$.
Then, the corresponding subfamily $\HH' = \{H_i\}_{i\in I}$ is a distance-$d$ $\HH$-packing in $G$ if and only if $\HH'$ is an independent $\HH$-packing in the graph $G^{d-1}$.
\end{observation}
\end{sloppypar}

\begin{proof}
Consider two distinct elements $i,j\in I$.
It suffices to show that $\dist_G(V(H_i),V(H_{j})) \ge  d$ if and only if $\dist_{G^{d-1}}(V(H_i),V(H_{j})) \ge 2$, or, equivalently, that $\dist_G(V(H_i),V(H_{j})) \le d-1$ if and only if $\dist_{G^{d-1}}(V(H_i),V(H_{j})) \le 1$.
Assume first that $\dist_G(V(H_i),V(H_{j})) \le d-1$ and let $P = (v_0,\ldots, v_r)$ be a path in $G$ from a vertex in $H_i$ to a vertex in $H_{j}$ such that $r\le d-1$.
If $r = 0$ then $\dist_{G^{d-1}}(V(H_i),V(H_{j}))= 0$.
If $r>0$, then in $G^{d-1}$, vertices $v_0$ and $v_{r}$ are adjacent, and hence $\dist_{G^{d-1}}(V(H_i),V(H_{j}))\le 1$.
Assume now that $\dist_{G^{d-1}}(V(H_i),V(H_{j})) \le 1$.
If this distance is~$0$, then the graphs $H_i$ and $H_{j}$ have a vertex in common and hence $\dist_G(V(H_i),V(H_{j})) = 0\le d-1$.
If this distance is $1$, then there is a path $P$ in $G$ of length at most $d-1$ from a vertex in $H_i$ to a vertex in $H_{j}$.
Thus, $\dist_G(V(H_i),V(H_{j}))\le d-1$.
\end{proof}

We now generalize \cref{max-weight-independent-subgraph-packing-for-bounded-tree-alpha}.

\begin{sloppypar}
\begin{theorem}\label{max-weight-distance-d-subgraph-packing-for-bounded-tree-alpha}
For every even positive integer $d$ and every $k\ge 1$, given a graph $G$, a finite family $\HH = \{H_j\}_{j\in J}$ of connected non-null subgraphs of $G$, and a weight function $\wei:J\to \mathbb{Q}_+$ on the subgraphs in $\HH$, the \textsc{Maximum Weight Distance-$d$ Packing} problem is solvable in time \hbox{$\mathcal{O}\left((|V(G)|+|V(T)|)\cdot|E(G)| + |J| \cdot |V(G)|\cdot ((|J| + |V(T)|+|V(G)|) + |J|^{k+1}\cdot|V(T)|)\right)$} if $G$ is given with a tree decomposition \hbox{$\mathcal{T} = (T, \{X_t\}_{t\in V(T)})$} with independence number at most~$k$.
\end{theorem}
\end{sloppypar}

\begin{sloppypar}
\begin{proof}
Let the input to the \textsc{Maximum Weight Distance-$d$ Packing} be given by the graph $G$, a finite family $\HH = \{H_j\}_{j\in J}$ of connected non-null subgraphs of $G$, and a weight function $\wei:J\to \mathbb{Q}_+$ on the subgraphs in $\HH$. 
The problem asks for a maximum-weight distance-$d$ $\HH$-packing in $G$, that is, a maximum-weight subfamily $\HH' = \{H_i\}_{i\in I}$ of subgraphs from $\HH$ that are at pairwise distance at least $d$ in $G$.
By \cref{observation}, for every $I\subseteq J$, the corresponding subfamily $\HH' = \{H_i\}_{i\in I}$ is a distance-$d$ $\HH$-packing in $G$ if and only if $\HH'$ is an independent $\HH$-packing in the graph $G^{d-1}$.

By \cref{tin-can-only-go-down-when-taking-odd-powers-algorithmic}, the graph $G^{d-1}$ and a tree decomposition \hbox{$\mathcal{T}' = (T, \{\Bag'_t\}_{t\in V(T)})$} of $G^{d-1}$ with independence number at most $k$ can be computed from $G$ and $\mathcal{T}$ in time $\mathcal{O}\left((|V(G)|+|V(T)|)\cdot(|V(G)|+|E(G)|)\right)$.
By \cref{max-weight-independent-subgraph-packing-for-bounded-tree-alpha}, a maximum-weight independent $\HH$-packing in the graph $G^{d-1}$ can be computed in time \hbox{$\mathcal{O}(|J| \cdot ((|J| + |V(T)|) \cdot |V(G^{d-1})| + |E(G^{d-1})|+|J|^{k}\cdot|V(T)|))$}.
Using the fact that $|V(G^{d-1})| = |V(G)|$ and $|E(G^{d-1})| \le |V(G)|^2$, the time complexity simplifies to 
\hbox{$\mathcal{O}(|J| \cdot ((|J| + |V(T)|) \cdot |V(G)| + |V(G)|^2+|J|^{k}\cdot|V(T)|))$}.
Altogether, we obtain the claimed time complexity $\mathcal{O}\left((|V(G)|+|V(T)|)\cdot|E(G)| + |J| \cdot |V(G)|\cdot ((|J| + |V(T)|+|V(G)|) + |J|^{k+1}\cdot|V(T)|)\right)$.  
\end{proof}
\end{sloppypar}

\medskip
Note that the proof of \cref{max-weight-distance-d-subgraph-packing-for-bounded-tree-alpha} does not require the distance bound $d$ to be fixed.
The result remains true, and with the stated time complexity, even if $d$ is part of input (as long as it is even).
Furthermore, if the input graph is not equipped with a tree decomposition of bounded independence number, we can obtain one in polynomial time using a recent result of Dallard et al.~\cite{dallard2022computing}.
Thus, given a graph $G$, a finite family $\HH = \{H_j\}_{j\in J}$ of connected non-null subgraphs of $G$, and a weight function $\wei:J\to \mathbb{Q}_+$ on the subgraphs in $\HH$, the \textsc{Maximum Weight Distance-$d$ Packing} problem is polynomial-time solvable in any class of graphs in which the tree-independence number is bounded by a  constant.

On the other hand, the assumption that $d$ is even is necessary in \cref{max-weight-distance-d-subgraph-packing-for-bounded-tree-alpha}, unless $\P= \NP$. 
Indeed, Eto et al.~\cite{MR3149106} showed that for all odd $d\ge 3$, the \textsc{Distance-$d$ Independent Set} problem is \NP-complete in the class of chordal graphs, which is exactly the class of graphs with tree-independence number at most~$1$.
Note that this \NP-completeness result is in line with the fact, shown in the proof of \cref{prop:no-bound}, that even powers of chordal graphs may contain arbitrary graphs as induced subgraphs.

\section{Finding large induced sparse subgraphs and \cmsotwo}\label{sec:cmso}
We start this section by recalling some basic notions and properties related to \cmsotwo.

\subsection{\cmsotwo: basic notions and properties}

\paragraph{\cmsotwo logic on graphs}
We assume that graphs are encoded as relational structures: each vertex and each edge is represented by a single variable (distinguishable by a unary predicate), and there is a single binary relation $\mathsf{inc}$ binding each edge with its endvertices.

\msotwo (\emph{Monadic Second Order}) logic is a logic on graphs that allows us to use vertex variables, edge variables, vertex sets variables, and edge set variables,
and, furthermore, to quantify over them. An example of an \msotwo formula is the following expression checking if the chromatic number of a given graph is at most $r$ (for constant $r$).
Here, $S_i$'s are vertex set variables, $e$ is an edge variable, and $x,y$ are vertex variables.
\begin{align}
\begin{split}
\phi_{\chi \leq r} \; :=  \;& \exists_{S_1,S_2,\ldots,S_r} \; (\forall_{x} \; x \in S_1 \cup S_2 \ldots \cup S_r) \; \land \\
 & \forall_{e} \forall_{x,y} \; (\mathsf{inc}((e,x) \land \mathsf{inc}(e,y) \land x \neq y) \to  \bigwedge_{i = 1}^r (x \notin S_i \lor y \notin S_i).	
 \end{split} \label{eq:chi-k}
\end{align}

For a positive integer $p$, by \cpmsotwo we mean the extension of \msotwo that allows us to use atomic formulae of the form $|S| \equiv a \bmod b$, where $S$ is a (vertex of edge) set variable, and $a,b \leq p$ are integers. We define \cmsotwo (\emph{Counting Monadic Second Order} logic) to be  $\bigcup_{p >0}$\cpmsotwo.
The \emph{quantifier rank} of a formula is the maximum number of nested quantifiers in the formula.

\paragraph{Boundaried graphs and  \cmsotwo types}
Let $\ell >0$ be an integer.
An \emph{$\ell$-boundaried graph} is a pair $(G,\iota)$, where $G$ is a graph and $\iota$ is a partial injective function from $V(G)$ to $[\ell]$. The domain of $\iota$, denoted by $\dom(\iota)$ is the \emph{boundary} of $(G,\iota)$. For $v \in \dom(\iota)$, the value of $\iota(v)$ is called the \emph{label} of $v$.

Let us define two natural operations on boundaried graphs.
For an $\ell$-boundaried graph $(G,\iota)$ and $l \in [\ell]$,
the result of \emph{forgetting the label $l$} is the $[\ell]$-boundaried graph $(G,\iota_{\neg l})$,
here $\iota_{\neg l} := \iota|_{\dom(\iota) \setminus \iota^{-1}(l)}$.
In other words, if the boundary of $(G,\iota)$ contains vertex $v$ with label $l$,
we remove the label from $v$, while keeping $v$ in the graph. Note that if $l \notin \dom(\iota)$, then
this operation does not do anything.

For two $\ell$-boundaried graphs $(G_1,\iota_1)$ and $(G_2,\iota_2)$,
the result of \emph{gluing} them is the $\ell$-boundaried graph $(G_1,\iota_1) \oplus_{[\ell]} (G_2,\iota_2)$ obtained from the disjoint union of $(G_1,\iota_1)$ and $(G_2,\iota_2)$ by identifying vertices with the same label.
Note that it might happen that some label  $l$ is used in, e.g., $\iota_1$ but not in $\iota_2$. Then the vertex with label $l$ in $(G_1,\iota_1) \oplus_{[\ell]} (G_2,\iota_2)$ is the copy of the vertex with label $l$ from $(G_1,\iota_1)$.

It is straightforward to observe that a graph has treewidth less than $\ell$
if and only if it can be constructed from $\ell$-boundaried graphs with at most two vertices by a sequence of forgetting and gluing operations (see, e.g.,~\cite{DBLP:journals/siamcomp/FominTV15}).

By \cmsotwo in $[\ell]$-boundaried graphs we mean \cmsotwo extended with $\ell$ unary predicates: for each $l \in [\ell]$ we have a predictate that selects a vertex with label $l$ (if such a vertex exists).
The following proposition is folklore (see, e.g.,~\cite[Lemma~6.1]{GroheK09}). The exact formulation we use here comes from Gartland et al.~\cite[Proposition~8]{DBLP:conf/stoc/GartlandLPPR21}.

\begin{proposition}\label{prop:mso-type}
For every triple of integers $\ell,p,q$, there exists a finite set $\msotypes^{\ell,p,q}$ and a function that assigns to every $\ell$-boundaried graph $(G,\iota)$ a \emph{type} $\msotype^{\ell,p,q}(G,\iota) \in \msotypes^{\ell,p,q}$ such that the following holds:
\begin{enumerate}
\item The types of isomorphic graphs are the same.
\item 
  For every \cpmsotwo sentence $\phi$ on $\ell$-boundaried graphs, whether $(G,\iota)$ satisfies $\phi$ depends only on the type $\msotype^{\ell,p,q}(G,\iota)$, where $q$ is the quantifier rank of $\phi$.
  More precisely,
  there exists a subset $\msotypes^{\ell,p,q}[\phi] \subseteq \msotypes^{\ell,p,q}$
  such that for every $\ell$-boundaried graph $(G,\iota)$ we have
  \[(G,\iota) \models \phi\qquad\textrm{if and only if}\qquad \msotype^{\ell,p,q}(G,\iota) \in \msotypes^{\ell,p,q}[\phi].\]
\item 
 The types of ingredients determine the type of the result of the gluing operation.
More precisely, for every two types $\tau_1,\tau_2 \in \msotypes^{\ell,p,q}$
there exists a type $\tau_1 \oplus_{\ell,p,q} \tau_2$ such that
for every two $\ell$-boundaried graphs $(G_1,\iota_1)$, $(G_2,\iota_2)$,
if $\msotype^{\ell,p,q}(G_i,\iota_i) = \tau_i$ for $i=1,2$, then
\[\msotype^{\ell,p,q}((G_1,\iota_1) \oplus_{[\ell]} (G_2,\iota_2)) = \tau_1 \oplus_{\ell,p,q} \tau_2.\]
Also, the operation $\oplus_{\ell,p,q}$ is associative and commutative.
\item 
 The type of the ingredient determines the type of the result of the forget label operation.
 More precisely, 
 for every type $\tau \in \msotypes^{\ell,p,q}$ and $l \in [\ell]$
there exists a type $\tau_{\neg l}$ such that
for every $\ell$-boundaried graph $(G,\iota)$,
if $\msotype^{\ell,p,q}(G,\iota) = \tau$
and $(G,\iota_{\neg l})$ is the result of forgetting $l$ in $(G,\iota)$, then
\[\msotype^{\ell,p,q}(G,\iota_{\neg l}) = \tau_{\neg l}.\]
\end{enumerate}
\end{proposition}

Intuitively, for any $\ell$-boundaried graph $(G,\iota)$,
whether $(G,\iota) \models \phi$ is fully determined by $\msotype^{\ell,p,q}(G,\iota)$ (where $p$ and $q$ are as in \cref{prop:mso-type}).
Furthermore, if $\phi$ and $\ell$ are fixed, then the sets $\msotypes^{\ell,p,q}$ and $\msotypes^{\ell,p,q}[\phi]$ are of constant size.
Finally, it is known that there is an algorithm that, given $\ell$ and $\phi$, computes $\msotypes^{\ell,p,q}$ and $\msotypes^{\ell,p,q}[\phi]$~\cite{DBLP:journals/algorithmica/BoriePT92}.

\subsection{\cmsotwo-definable subgraphs of graphs with small tree-independence number}

Now we are ready to prove our main result concerning sparse induced subgraphs satisfying some \cmsotwo-definable property.
\begin{theorem}\label{thm:cmsomain}
For every $k,r$ and a \cmsotwo formula $\psi$ there exists a positive integer $f(k,r,\psi)$ such that the following holds.
Let $G$ be a graph given along with a tree decomposition of independence number at most $k$,
and let $\wei : V(G) \to \Q_+$ be a weight function.
Then in time $f(k,r,\psi) \cdot |V(G)|^{\Oh(kr)}$ we can find a set $Y \subseteq V(G)$, such that
\begin{enumerate}
\item $G[Y] \models \psi$,
\item $\chi(G[Y]) \leq r$,
\item $F$ is of maximum weight subject to the conditions above,
\end{enumerate}
or conclude that no such set exists.
\end{theorem}

\begin{proof}
Let $\cT=(T,\{X_t\}_{t \in V(T)})$ be a tree decomposition of $G$ with independence number at most $k$.
Recall that without loss of generality we can assume that $\cT$ is nice.
Let $F$ be a fixed (unknown) optimal solution. Note that by removing all vertices from $V(G) \setminus F$ from each bag of $\cT$, we obtain a tree decomposition of $G[F]$ with independence number at most $k$. As $\chi(G[F]) \leq r$, we conclude that every bag of $\cT$ contains at most $kr$ vertices of $F$.
In particular, the treewidth of $G[F]$ is less than $\ell:=kr$.
This is the only way in which we use that $\chi(G[F]) \leq r$.
Thus, in order to make sure that the obtained solution satisfies the bound on the chromatic number,
instead of $\psi$ we will use the \cmsotwo formula $\phi := \psi \land \phi_{\chi \leq r}$, where $\phi_{\chi \leq r}$ is the formula \eqref{eq:chi-k}.

We will think of $G[F]$ as constructed by a sequence of forgetting and gluing operations on $\ell$-boundaried graphs.

Consider the minimum integer $p$ such that $\phi$  belongs to \cpmsotwo, and let $q$ be the quantifier rank of $\phi$.
Let $\msotypes^{\ell,p,q}$ and $\msotypes^{\ell,p,q}[\phi]$ be as in \cref{prop:mso-type}.
Recall that these sets are of constant size (since $k,r$, and $\phi$ are fixed); we assume that they are hard-coded in our algorithm.
Similarly, we assume that the forgetting and gluing operations on types, i.e., functions 
\[(\tau_1,\tau_2) \mapsto \tau_1 \oplus_{\ell,p,q} \tau_2\qquad\textrm{and}\qquad (\tau,l) \mapsto \tau_{\neg l},\]
are hard-coded. Finally, our algorithm has the list of types of all $\ell$-boundaried graphs with at most $\ell$ vertices.
As $\ell,p,q$ are fixed throughout the algorithm, we will simplify the notation by dropping subscripts and superscripts, i.e., we will simply write $\msotypes, \msotypes[\phi]$, and $\oplus$.

Our algorithm is a standard bottom-up dynamic programming on a nice tree decomposition. We will describe how to compute the weight of an optimal solution; adapting our approach to return the solution itself is a standard task.

We introduce the table $\Tab[\cdot,\cdot]$ indexed by nodes $t$ of $T$ and \emph{states}.
Each state is a triple $(S, \iota, \tau)$, where $S \subseteq X_t$ has at most $\ell$ vertices,
$\iota$ is an injective function from $S$ to $[\ell]$, and $\tau \in \msotypes$. 
The value of $\Tab[t,(S,\iota,\tau)]$ will encode the maximum weight of a subset $F_t$ of $V(G_t)$ such that
\begin{enumerate}
\item $F_t \cap X_t = S$,
\item $\msotype((G[F_t],\iota))=\tau$.
\end{enumerate}
Intuitively, we think of $S$ as the intersection of a fixed optimal solution with $X_t$ and we look for a maximum-weight induced $\ell$-boundaried subgraph of $G_t$ with boundary $S$ and type $\tau$.
If there is no subset $F_t$ satisfying the conditions above, we will have $\Tab[t,(S,\iota,\tau)]=-\infty$. 
Clearly, having computed all entries of $\Tab[\cdot,\cdot]$, the weight of the sought-for optimum solution is
\[
\max \Tab[\mathrm{root}(T), (S,\iota,\tau)] \text{ over all feasible } S, \iota \text{ and } \tau \in \msotypes[\phi],
\]
where by $\mathrm{root}(T)$ we mean the root of $T$.

Let us now describe how we compute the entries $\Tab[t,\cdot]$ for a node $t$ of $T$.
The exact way depends on the type of $x$.

\paragraph*{Case: $t$ is a leaf node.}
Note that in this case $G_t$ has only one vertex, and thus the values of $\Tab[t,\cdot]$ can be extracted from the information that is hard-coded in the algorithm.

\paragraph*{Case: $t$ is an introduce node.}
Let $t'$ be the child of $t$ in $T$, and let $v$ be the unique vertex in $X_t \setminus X_{t'}$.
Clearly if $v \notin S$, then $\Tab[t,(S,\iota,\tau)] = \Tab[t,(S,\iota,\tau)]$.

So suppose $v \in S$, let $S' = S \setminus \{v\}$, and let $\iota' = \iota|_{S'}$.
Let $G_S$ be the subgraph of $X_t$ induced by $S$. As $|S| \leq \ell$, the value $\tau_S := \msotype((G_S,\iota))$ is hard-coded in our algorithm.
Now we set
\[
\Tab[t,(S,\iota,\tau)] = \max_{\tau'} \Tab[t',(S',\iota',\tau')],
\]
where the maximum is taken over all $\tau'$ such that $\tau_S \oplus \tau' = \tau$.

\paragraph*{Case: $t$ is a forget node.}
Let $t'$ be the child of $t$ in $\cT$, and let $v$ be the unique vertex in $X_{t'} \setminus X_t$.
Intuitively, partial solutions at $t$ correspond to partial solutions at $t'$ with the label of $v$ (if any) forgotten.

Thus 
\[
\Tab[t,(S,\iota,\tau)] = \max \left( \Tab[t',(S,\iota,\tau)] , \max_{\iota',\tau'} \Tab[t',(S \cup \{v\},\iota',\tau')] \right),
\]
where in the second term we assume that $|S| < \ell$.
The maximum in the second term is computed over all pairs $\iota',\tau'$, such that $\iota'|_{S}=\iota$ and $\tau'|_{\neg \iota'(v)}=\tau$.

\paragraph*{Case: $t$ is a join node.}
Let $t',t''$ be the children of $t$.
Thus 
\[
\Tab[t,(S,\iota,\tau)] = \max_{\tau',\tau''} \Tab[t',(S,\iota,\tau')] + \Tab[t'',(S,\iota,\tau'')] - \wei(S),
\]
where the maximum is taken over all $\tau',\tau''$ such that $\tau' \oplus \tau'' = \tau$. Note that we subtract the weight of $S$,
as it is counted both for $t'$ and for $t''$.

The correctness of the procedure above follows directly from the properties of tree decompositions and \cref{prop:mso-type}.

For each node $t$ of $T$, the number of states is at most $|X_t|^{\ell} \cdot \ell! \cdot |\msotypes|$.
Each entry of $\Tab[\cdot,\cdot]$ can be computed in time proportional to the number of states for each node (here we use that $k,r,\phi$ are fixed).
As the number of nodes of $T$ is polynomial in $|V(G)|$, we conclude that the running time is $f(k,r,\psi) \cdot |V(G)|^{\Oh(kr)}$ for some function $f$ depending on $k,r$, and $\psi$.
\end{proof}

\subsection{Generalizations}

In the first generalization we show that we can maximize the weight of some carefully chosen subset of the solution.
The formula $\psi$ uses this set as a free variable.

\begin{theorem}\label{thm:cmsomainX}
For every $k,r$ and a \cmsotwo formula $\psi$ with one free vertex-set variable, there exists a positive integer $f(k,r,\psi)$ such that the following holds.
Let $G$ be a graph given along with a tree decomposition of independence number at most $k$,
and let $\wei : V(G) \to \Q_+$ be a weight function.
Then in time $f(k,r,\psi) \cdot |V(G)|^{\Oh(kr)}$ we can find sets $X \subseteq F \subseteq V(G)$, such that
\begin{enumerate}
\item $(G[F],X) \models \psi$,
\item $\chi(G[F]) \leq r$,
\item $X$ is of maximum possible weight.
\end{enumerate}
\end{theorem}
\begin{proof}[Proof sketch] The proof follows by a simple trick used by Gartland et al.~\cite[Theorem 38]{DBLP:conf/stoc/GartlandLPPR21}.
We just sketch the argument here and refer the reader to~\cite[Theorem 38]{DBLP:conf/stoc/GartlandLPPR21} for more details.

For a graph $G$ and a subset $M \subseteq V(G)$, by $\forked{G}$ we denote the \emph{forked version of} $(G,M)$, which is obtained as follows.
We start the construction with a copy of $G$.
For every vertex $v$ of $G$ we introduce to $\forked{G}^M$ three vertices of degree 1, each adjacent only to $v$.
For every vertex $v$ of $M$ we introduce to $\forked{G}^M$ a two-edge path, whose one endvertex
is identified with $v$.
If $G$ is equipped with weights $\wei$, we set the weights $\forked{\wei}$ of all vertices of $\forked{G}^M$ to 0, with the exception that the value of $\wei(v)$ for $v \in M$ is transferred to the other endvertex of the two-edge path attached to $v$.

Note that if $F'$ is a forked version of some $(F,M)$, then $F$ and $M$ can be uniquely decoded from $F'$.
Indeed, $F$ is the subgraph of $F'$ induced by vertices of degree at least four, and vertices of $M$ are those that are adjacent (in $F'$) to a vertex of degree 2. Furthermore, $\chi(F') \leq \max(\chi(F'),1)$.

Moreover, observe that if $G$ is given with a tree decomposition $\cT=(T,\{X_t\}_{t \in V(T)})$ with independence number at most $k$,
then we can easily obtain a tree decomposition $\forked{\cT}=(\forked{T},\{\forked{X}_t\}_{t \in V(\forked{T})})$ of $\forked{G}^{V(G)}$ with independence number at most $k$.
The idea is to call the algorithm from \cref{thm:cmsomain} on $\forked{G}^{V(G)}$ and the formula $\forked{\psi}$ that for a graph $F'$ ensures that
\begin{itemize}
\item if $F'$ is a forked version of some $(F,M)$, then $F' \models \forked{\psi}$ if and only if $F \models \psi$,
\item otherwise $F' \not\models \forked{\psi}$.
\end{itemize}
Such a formula can be easily constructed in \cmsotwo~\cite{DBLP:conf/stoc/GartlandLPPR21}.
We observe that the solution found by the algorithm satisfies the conditions stated in the theorem; again, see~\cite[Theorem 38]{DBLP:conf/stoc/GartlandLPPR21} for details.
\end{proof}

The idea introduced in the proof of \cref{thm:cmsomainX} allows us to annotate vertices.
In other words, our input is a graph and a constant number of subsets of vertices. The formula $\psi$ uses these sets as free variables,
and we assume that the instance graph $G$ is given along with a  valuation of these free variables.

\begin{theorem}\label{thm:cmsomainAnn}
For every $k,r,p$ and a \cmsotwo formula $\psi$ with $p+1$ free vertex-set variables there exists a positive integer $f(k,r,p,\psi)$ such that the following holds.
Let $G$ be a graph given along with a tree decomposition of independence number at most $k$,
and let $\wei : V(G) \to \Q_+$ be a weight function.
Furthermore, let $A_1,A_2,\ldots,A_p \subseteq V(G)$ be specified subsets of vertices of $G$.
Then in time $f(k,r,p,\psi) \cdot |V(G)|^{\Oh(kr)}$ we can find sets $X \subseteq F \subseteq V(G)$, such that
\begin{enumerate}
\item $(G[F],X,A_1,A_2,\ldots,A_p) \models \psi$,
\item $\chi(G[F]) \leq r$,
\item $X$ is of maximum possible weight.
\end{enumerate}
\end{theorem}

\cref{thm:cmsomainAnn} can be proven in a way analogous to \cref{thm:cmsomainX}, but this time we distinguish the vertices from the sets $A_i$ by the number of attached leaves. We skip the details, as they do not bring any new insight.

\subsection{Algorithmic applications}
Using the framework of \cref{thm:cmsomain} we can solve several well-studied problems in graphs given with a tree decomposition of bounded independence number:
\begin{itemize}
\item \textsc{Max Weight Independent Set},
\item \textsc{Max Weight Induced Forest}, which is equivalent by complementation to \textsc{Min Weight Feedback Vertex Set},
\item \textsc{Max Weight Induced Bipartite Subgraph}, which is equivalent by complementation to \textsc{Min Weight Odd Cycle Transversal},
\item \textsc{Max Weight Induced Odd Cactus}, which is equivalent by complementation to \textsc{Min Weight Even Cycle Transversal},
\item \textsc{Max Weight Induced Planar Subgraph}, which is equivalent to \textsc{Planarization}.
\end{itemize}

Using the extension from \cref{thm:cmsomainX} we can solve some other problems, for example:
\begin{itemize}
\item find the maximum number of pairwise independent induced cycles.
\end{itemize}
Note that the difference is that we maximize the \emph{number} of cycles, and not their \emph{total size} (or weight), as we could do using \cref{thm:cmsomain}.
On other other hand, the number of induced cycles in a graph might be very large, so this problem cannot be solved using \cref{max-weight-independent-subgraph-packing-for-bounded-tree-alpha} (or \cref{max-weight-distance-d-subgraph-packing-for-bounded-tree-alpha}).
However, we can solve our problem by calling \cref{thm:cmsomainX} with the following parameters:
\begin{itemize}
    \item $d=3$,
    \item $\psi$ is the formula saying that $G[F]$ is 2-regular (i.e., each component of $G[F]$ is a cycle) and each vertex from $X$ must be in a distinct connected component of $G[F]$, and
    \item the weight of each vertex is 1.
\end{itemize}

Finally, the extension from \cref{thm:cmsomainAnn} allows us to solve problems where we allow annotations on vertices.
An example of such a problem is the following one:
\begin{itemize}
\item given a graph, whose every vertex is equipped with a subset of $[r]$ (called a \emph{list}),
find a maximum induced subgraph that admits a proper $r$-coloring respecting the lists.
\end{itemize}
Indeed, it is sufficient to call \cref{thm:cmsomainAnn}, where each set $A_i$ contains vertices with a given list (note that there are at most $2^r$, i.e., a constant number of distinct lists), and $\psi$ is a formula similar to \eqref{eq:chi-k}, but also checking if the colors assigned to vertices respect the lists.  Note that the last problem is a generalization of \textsc{List}-$r$-\textsc{Coloring} (for constant $r$).

\medskip
Furthermore, let us recall that if the input graph is not equipped with a tree decomposition with bounded independence number, we can again invoke a result from~\cite{dallard2022computing} to infer that all of the above problems are polynomial-time solvable in any class of graphs in which the tree-independence number is bounded by a  constant.

\medskip
In conclusion, let us point out that in \cref{thm:cmsomain,thm:cmsomainX,thm:cmsomainAnn} we need to assume that the chromatic number $d$ of the solution is constant. Indeed, if the number of colors is a part of the input, already the (\textsc{List-}) \textsc{Coloring} problem becomes \NP-hard in graphs of bounded tree-independence number.

\begin{theorem}The following problems are \NP-hard:
\begin{enumerate}
    \item \textsc{List-Coloring} in graphs of tree-independence number 1,
    \item \textsc{Coloring} in graphs of tree-independence number 2.
\end{enumerate}
\end{theorem}
\begin{proof}
The first statement follows from the result of Golovach and Paulusma~\cite{DBLP:journals/dam/GolovachP14}, who showed that \textsc{List-Coloring} is \NP-hard in complete split graphs, i.e., graphs whose vertex set can be partitioned into an independent set and a clique, and all edges between these two sets are present. Complete split graphs are in particular chordal
and thus of tree-independence number $1$.


To show the second statement, it suffices to observe that the \textsc{Coloring} problem is $\NP$-complete in the class of circular-arc graphs, as shown by Garey et al.~\cite{MR578325}.
Furthermore, as we explain next, circular-arc graphs have tree-independence number at most two. 
Every circular-arc graph $G$ is the intersection graph of a family $\mathcal{F}$ of connected subgraphs of a graph $H$ such that $H$ is a cycle.
Fix a vertex $v\in V(H)$ and let $G_{\bar{v}}$ denote the intersection graph of the graphs in $\mathcal{F}_{\bar{v}} = \{F\in \mathcal{F}: v\not\in V(F)\}$.
Then $G_{\bar{v}}$ is a chordal graph (in fact, an interval graph) and hence admits a tree decomposition $\mathcal{T}$ in which each bag is a clique.
Finally, denoting by $K$ the set of vertices of $G$ corresponding to graphs in $\mathcal{F}$ containing $v$, a tree decomposition of $G$ with independence number at most $2$ can be obtained from the tree decomposition $\mathcal{T}$ of $G_{\bar{v}}$ by adding the vertices of the clique $K$ to every bag.
\end{proof}

\subsection*{Acknowledgements}
The authors are grateful to O-joung Kwon for asking about the complexity of the \textsc{Distance-$d$ Independent Set} problem in classes of graphs with bounded tree-independence number, and to Cl\'ement Dallard and Kenny Štorgel for helpful discussions and remarks. 
This work is supported in part by the Slovenian Research Agency (I0-0035, research program P1-0285 and research projects J1-9110, N1-0102, N1-0160, J1-3001, J1-3002, and J1-3003).

\bibliographystyle{abbrv}
\bibliography{biblio}
\end{document}